\documentclass[authorcolumns,numberwithinsect,a4paper]{no-lipics}

\newcommand{\Oh}[1]{\mathcal{O}(#1)}
\newcommand{\pspace}{\ensuremath{\mathsf{PSPACE}}}

\newcommand{\iso}{\cong}
\newcommand{\incG}[1]{\mathcal{H}(#1)}
\newcommand{\claws}[2][]{\mathsf{Cw}_{#1}(#2)}
\newcommand{\concomp}[2]{\kappa(#1,#2)}
\usepackage{cleveref}

\title{A Linear-Time Algorithm for the Unbiased Claw Game on Graphs with Maximum Degree~4}
\titlerunning{A Linear-Time Algorithm for the Unbiased Claw Game on Graphs \boldmath\({G}\) with \(\Delta(G)=4\)}
\author{Julian Christoph Brinkmann}{Goethe University Frankfurt, Germany}{J.Brinkmann@em.uni-frankfurt.de}{https://orcid.org/0009-0000-0332-4543}{}
\authorrunning{J. C. Brinkmann}
\begin{document}

\thispagestyle{plain}
\maketitle
\begin{abstract}
	We study the unbiased Maker-Breaker \(K_{1,3}\)-game, the unbiased claw game, played on the edge set of a general graph \(G\).
	In this game, Maker and Breaker alternatingly claim edges of \(G\), and Maker wins if his claimed edges contain a claw once no more edges are available.
	Graphs with maximum degree at least~5 are Maker's win.
	A graph \(G\) is claw-connected if the hypergraph modeling the unbiased claw game played on \(G\) is connected.
	We characterize Maker's win in the unbiased claw game on every claw-connected graph whose maximum degree is different from~3.
	Our main structural result is a complete characterization of Maker's win for claw-connected graphs with exactly one vertex of degree~4.
	The analysis is based on the fact that vertices of degree~4 act as forcing structures.
	
	Our characterization results in a linear-time algorithm for the unbiased claw game on claw-connected graphs with maximum degree different from~3.
	We also obtain a linear-time decomposition algorithm reducing the game on general graphs to a collection of pairwise edge-disjoint claw-connected subgraphs of maximum degree exactly~3, thereby isolating the only remaining case.
	This decomposition yields an \(\Oh{|V(G)|^3}\)-time algorithm for the unbiased claw game on general graphs, improving on the \(\Oh{|V(G)|^7}\)-time algorithm implied by the work of Galliot, Gravier and Sivignon~(arXiv~2026).
	For a lower bound on the decision complexity, we prove that any agent deciding the game in the adjacency-array model makes at least \(\Omega(|V(G)|)\) queries in the worst case, indicating the optimality of our linear-time algorithm.
\end{abstract}

\setcounter{tocdepth}{2}
\tableofcontents
\clearpage

\section{Introduction}
A positional game is a two-player game in which the players take turns claiming vertices of a hypergraph.
There are several variants with different winning conditions, including Maker-Maker, Maker-Breaker, and Avoider-Enforcer games; see~\cite{hefetz2014positional} for an introduction to the area.
We focus on Maker-Breaker games, which were introduced by Erdős and Selfridge~\cite{DBLP:journals/jct/ErdosS73}, where Maker aims to occupy every vertex of a hyperedge while Breaker seeks to stop Maker from doing so.
A common setting is to define the game on the edge set of an ordinary graph, treating its edges as hypergraph vertices and declaring certain collections of them winning, for example those edge sets forming a spanning tree, perfect matching, Hamiltonian cycle or a copy of a \emph{fixed} graph \(H\).
In this work, we investigate the Maker-Breaker \(K_{1,3}\)-game, the unbiased claw game, on general graphs.
In the unbiased game, each player claims exactly one vertex in every round.
For biased variants, in which Breaker is allowed to claim several vertices in each round, see~\cite{chvatal1978biased,DBLP:journals/combinatorica/BednarskaL00}.

The computational complexity of Maker-Breaker games on hypergraphs has recently been fully classified.
A hypergraph has rank \(k\) if all of its hyperedges contain at most \(k\) vertices and at least one hyperedge contains exactly \(k\) vertices.
In a seminal work, Schaefer~\cite{DBLP:journals/jcss/Schaefer78} proved that Maker-Breaker games are \(\pspace\)-complete even when restricted to hypergraphs with rank most~11.
In subsequent works, the rank bound was reduced to~6 by Rahman and Watson~\cite{DBLP:journals/combinatorica/RahmanW23}, to rank~5 by Koepke~\cite{DBLP:journals/corr/abs-2502-20271} and finally to rank~4 by Galliot and Sénizergues~\cite{DBLP:conf/stacs/GalliotS26}.
Galliot~\cite{DBLP:journals/corr/abs-2509-13819} also obtained the stronger result that \(\pspace\)-completeness already holds for 4-uniform hypergraphs, in which every hyperedge consists of exactly four vertices.
In contrast, Maker-Breaker games on hypergraphs of rank at most~3 admit a polynomial-time algorithm for determining the winner of the game~\cite{galliot2025makerbreakersolvedpolynomialtime}.
This extends the earlier result of Kutz~\cite{kutz2004thesis}, who obtained a polynomial-time algorithm for linear hypergraphs of rank at most~3.
A hypergraph is linear if two distinct hyperedges intersect in at most one vertex.

These results all deal with the complexity of unrestricted Maker-Breaker games on hypergraphs.
Duchéne et al.~\cite{DBLP:journals/dam/DucheneGINOPS25} investigated the complexity of the Maker-Breaker subgraph game, the \(H\)-game.
They proved that the Maker-Breaker \(T\)-game is \(\pspace\)-complete even for board graphs \(G\) with diameter at most~6, where \(T\) is a fixed tree on~91 vertices.
They also established \(\pspace\)-completeness of the \(H\)-game for a particular graph \(H\) with~51 vertices and~57 edges.
Their constructions rely on the hardness of Maker-Breaker games on hypergraphs of rank~6.
The improvement from~\cite{DBLP:journals/corr/abs-2509-13819} allows these constructions to be reduced to a tree \(T\) with~45 vertices and a graph \(H\) with~35 vertices and~39 edges, respectively.

For most target graphs \(H\), the complexity of the corresponding Maker-Breaker \(H\)-game is open.
Galliot~\cite{DBLP:journals/corr/abs-2509-13819} proved \(\pspace\)-completeness of the \(C_4\)-game played on the \emph{vertices}, rather than the edges, of general graphs.
The complexity of the usual edge-based \(C_4\)-game remains open.
The polynomial-time algorithm of Galliot, Gravier and Sivignon~\cite{galliot2025makerbreakersolvedpolynomialtime} yields polynomial-time solvability of the \(H\)-game if \(H\) has at most three edges.
If \(H\) has at most two edges, it is easy to determine the winner of the \(H\)-game in linear time.
For disconnected \(H\), Maker is essentially playing different games in parallel.
Thus, connected three-edge graphs are of particular interest.
The \(P_4\)-game, where \(P_4\) denotes the four-vertex path, was considered in~\cite{DBLP:journals/dam/DucheneGINOPS25}.
The triangle game was analyzed in~\cite{DBLP:journals/corr/abs-2607-18569}.
The only remaining connected three-edge target graph is the claw, which received only limited attention in~\cite{DBLP:journals/dam/DucheneGINOPS25}, where the authors established an FPT-algorithm for the \(K_{1,3}\)-game.

The algorithm presented in~\cite{galliot2025makerbreakersolvedpolynomialtime} has running time \(\Oh{|V(\mathcal{H})|^5|E(\mathcal{H})|^2+|V(\mathcal{H})|^6\Delta(\mathcal{H})}\), where \(\Delta(\mathcal{H})\) denotes the maximum vertex degree in the hypergraph \(\mathcal{H}\).
While this result does provide a polynomial-time algorithm for determining the outcome of the \(H\)-game whenever \(H\) contains at most three edges, the asymptotic running time is \(\Oh{|V(G)|^{16}}\) when expressed in terms of the board graph \(G\) as the vertices of \(\mathcal{H}\) are the edges of \(G\) and the hyperedges are the edge-induced subgraphs of \(G\) that are isomorphic to \(H\).
Duchéne et al.~\cite{DBLP:journals/dam/DucheneGINOPS25} characterized the maximal board graphs on which Breaker wins the \(P_4\)-game.
Since membership in this family of maximal Breaker-winning graphs can be tested in linear time, this yields a linear-time algorithm for the \(P_4\)-game.
For the unbiased triangle game, Brinkmann and Srivastav~\cite{DBLP:journals/corr/abs-2607-18569} obtained an algorithm with running time \(\Oh{|V(G)|^7}\).
Thus, exploiting properties specific to the target graph \(H\) can lead to considerably better algorithms than the general-purpose algorithm of~\cite{galliot2025makerbreakersolvedpolynomialtime}.

Duchéne et al.~\cite{DBLP:journals/dam/DucheneGINOPS25} also show that the \(K_{1,\ell}\)-game is fixed-parameter tractable when parameterized by the length of the game, where \(K_{1,\ell}\) is a star graph with \(\ell\) edges.
The authors of~\cite{DBLP:journals/dam/DucheneGINOPS25} do not give a concrete bound for the running time of their FPT-algorithm.
No other attempts have been made to improve on the \(\Oh{|V(G)|^{16}}\) algorithm of~\cite{galliot2025makerbreakersolvedpolynomialtime} for the unbiased claw game.
Arguably, an FPT-algorithm, especially without a concrete running time, is no improvement at all compared to a polynomial-time algorithm.

\subsection{Our Contribution}
We close this gap by making the first step towards a linear-time algorithm for the unbiased claw game, giving such an algorithm for all claw-connected graphs with maximum degree different from~3.
A graph \(G\) is claw-connected if the hypergraph modeling the unbiased claw game on \(G\) is connected, see also \cref{defn:claw_connected}.
Further, we give an algorithm for the unbiased claw game on arbitrary graphs with running time \(\Oh{|V(G)|^3}\).
This improves the state of the art by a factor of \(|V(G)|^4\) even when adding the domain-specific improvement given by \cref{prop:deg5_trivial} to the algorithm from \cite{galliot2025makerbreakersolvedpolynomialtime}, see \cref{fig:claw_alg_table}.

Additionally, we give a decomposition algorithm that allows future work to focus only on claw-connected graphs with maximum degree exactly~3 to settle the problem, see \cref{cor:decomp_alg}.
Finally, we show that, in the adjacency-array model of graphs, where the degrees of all vertices are given and the adjacency array of all vertices can be queried, every computationally unbounded agent deciding the unbiased claw game makes at least \(\Omega(|V(G)|)\) queries.
This indicates that a linear running time is best possible within standard models of computation.

\begin{figure}
	\centering
	\begin{tabular}{c|l|c}
		Graph class & \multicolumn{1}{c|}{Running time} & Source\\\hline
		\multirow{4}{*}{general graphs} & \(f(k)\cdot\text{poly}(|V(G)|)\) & \cite{DBLP:journals/dam/DucheneGINOPS25}\\
		& \(\Oh{|V(G)|+|E(G)|^5|\claws{E(G)}|^2+|V(G)||E(G)|^6}\) &\cite{galliot2025makerbreakersolvedpolynomialtime}\\
		& \(\Oh{|V(G)|^7}\) &\cite{galliot2025makerbreakersolvedpolynomialtime} with \cref{prop:deg5_trivial}\\
		& \(\Oh{|V(G)|^3}\) & this work\\\hline
		claw-connected, \(\Delta(G)\neq3\)& \(\Oh{|V(G)|}\) & this work
	\end{tabular}
	\caption{Overview of algorithms for the unbiased claw game on different graph classes.
	The authors of~\cite{DBLP:journals/dam/DucheneGINOPS25} do not provide a concrete running time bound.}
	\label{fig:claw_alg_table}
\end{figure}

From a technical perspective, we augment the framework of Galliot, Gravier and Sivignon \cite{galliot2025makerbreakersolvedpolynomialtime}, in which Maker's win is characterized by structures arising within at most three full rounds of play, by limiting player choices to a threatening structure for Breaker.
The threatening structure is given by vertices with degree~4:
When Maker claims an edge incident to a vertex \(v\) with degree~4, Breaker must also claim an edge incident to \(v\) to avoid losing.
This is similar to the role played by subgraphs isomorphic to \(K_4\), the complete graph on four vertices, in the unbiased triangle game, see~\cite[Theorem~4.17]{DBLP:journals/corr/abs-2607-18569}.
As the algorithm from~\cite{galliot2025makerbreakersolvedpolynomialtime} requires the simulation of three full rounds of play, a running time faster than \(\Oh{|E(G)|^6}\) cannot be achieved using their methods.
However, our approach ultimately allows us to reduce to a situation in which the game can be modeled by a \emph{linear} hypergraph.
This enables the use of more efficient algorithms instead of \cite{galliot2025makerbreakersolvedpolynomialtime} for deciding the winner of the game, leading to \cref{thm:claw_con_alg}.

The paper is structured as follows:
In \cref{sec:claw_prelims}, we introduce the terminology used in this work.
Then, we characterize Maker's win in the unbiased claw game played on claw-connected graphs with maximum degree at least~4 in \cref{sec:analysis}.
\Cref{sec:alg} contains our algorithmic results.
Finally, we show a lower bound on the query complexity of the problem in \cref{sec:query_lb}.

\section{Preliminaries}\label{sec:claw_prelims}
\subparagraph*{Graphs.}
Let \(G\) be a graph.
We denote the vertex set of \(G\) by \(V(G)\) and the edge set of \(G\) by \(E(G)\).
We denote the \emph{neighborhood} of \(v\in V(G)\) in \(G\) by \(N_G(v)\).
Formally, \(N_G(v)=\{w\in V(G): uw\in E(G)\}\).
The \emph{degree} of a vertex is \(|N_G(u)Y\).
The maximum degree of a vertex in \(G\) is denoted \(\Delta(G)\), i.e.~\({\Delta(G)=\max\{|N_G(v)|:v\in V(G)\}}\).
The graph \(K_{a,b}\) is the complete bipartite graph where the two color classes have size \(a\) and \(b\), respectively.
We call (copies of) \(K_{1,3}\) \emph{claws}.
We write \(G\iso H\) to denote that the graphs \(G\) and \(H\) are isomorphic.
For \(F\subseteq E(G)\), the subgraph induced by \(F\) is denoted by \(G[F]=(V_F,F)\) where \(V_F\) is the set of endpoints of edges in \(F\).
We use \(G-F\) to denote the graph \((V(G),E(G)\setminus F)\).
For edges \(e\in E(G)\), we write \(G-e\) instead of \(G-\{e\}\).
The \emph{set of claws} using edges from \(F\) is denoted by \(\claws[G]{F}\).
Formally, \(\claws[G]{F}=\{S\subseteq F: G[S]\iso K_{1,3}\}\).
For \(U\subseteq V(G)\), the subgraph induced by \(U\) is denoted by \(G[U]=(U,E_U)\) where \(E_U=\{e\in E(G)\mid e\subseteq U\}\).
We use \(G-U\) to denote the induced subgraph \(G[V(G)\setminus U]\).
For a vertex \(v\in V(G)\), we write \(G-v\) instead of \(G-\{v\}\).
We will omit the subscript when \(G\) is clear from the context.

\subparagraph*{Hypergraphs.} We use the same notation for hypergraphs.
Given a hypergraph \(H\) and a subset \({U\subsetneq V(H)}\), the hypergraph \(H^{+U}\) has vertex set \(V(H)\setminus U\) and edge set \({\{e\setminus U:e\in E(H)\}\setminus\{\emptyset\}}\).
For \(v\in V(H)\), we write \(H^{+v}\) instead of \(H^{+\{v\}}\).
For \(\emptyset\neq e\subseteq V(H)\), the hypergraph \(H+e\) has vertex set \(V(H)\) and edge set \(E(H)\cup\{e\}\).

The unbiased claw game played on the graph \(G\) is the following two player game:
The players Maker and Breaker alternatingly claim edges of \(G\) starting with Maker.
In each round, Maker and Breaker claim exactly one edge that has not already been claimed by any player as long as there are still edges available.
After all edges are claimed, Maker wins if a subset of the edges he claimed induces a \(K_{1,3}\) in \(G\).
Otherwise, Breaker wins.
If Maker wins the unbiased claw game on \(G\), \(G\) is called Maker's win.
Otherwise, \(G\) is called Breaker's win.
This game is modeled by the following hypergraph:

\begin{definition}\label{defn:claw_connected}
	Let \(G\) be a graph.
	We define the hypergraph \(\incG{G}\) by
	\begin{align*}
		V(\incG{G})=E(G) && E(\incG{G})=\claws[G]{E(G)}.
	\end{align*}
	If \(\incG{G}\) is connected, then \(G\) is called \emph{claw-connected}.
\end{definition}

Two different claws can intersect in up to two edges when \(G\) contains a vertex with degree at least~4, so the hypergraph \(\incG{G}\) is not necessarily linear.
Every claw-connected graph without isolated vertices is connected.
However, not every connected graph is claw-connected, e.g.~a path consisting of two edges.

\subparagraph*{Marked hypergraphs.}
Some results from \cite{galliot2025makerbreakersolvedpolynomialtime} are required for our analysis.
We now introduce the necessary terminology.
A \emph{marked hypergraph} is a pair \((H,M)\), where \(H\) is a hypergraph and \({M\subseteq V(H)}\) is the set of \emph{marked vertices}.
During a Maker-Breaker game, the updated marked hypergraph is obtained by marking the vertices claimed by Maker and deleting the vertices claimed by Breaker.
For example, if Maker claims \(v_1\) and Breaker claims \(v_2\) in the current marked hypergraph \((H,M)\), then the updated marked hypergraph is \((H-v_2,M\cup\{v_1\})\).
We usually view the updated hypergraph as \(H^{+M\cup\{v_1\}}-v_2\), but the two views are equivalent.

Let \(H\) be a 3-uniform hypergraph.
A \emph{chain} in \(H\) is a sequence \(P=(e_1,\dots,e_k)\) of pairwise distinct hyperedges such that \(|e_i\cap e_{i+1}|=1\) for every \(i\) such that \(1\leq i<k\) and \(e_i\cap e_j=\emptyset\) whenever \(|i-j|\geq2\).
A \emph{cycle} is defined analogously by further requiring \(|e_1\cap e_k|=1\).
Additionally, a pair of hyperedges \((e_1,e_2)\) is also called a cycle if \(|e_1\cap e_2|=2\).
The inner vertices of a cycle \((e_1,\dots,e_k)\) are the vertices with degree~2 in \(H[\{e_1,\dots,e_k\}]\).
An \emph{\(xy\)-nunchaku} is a chain \((e_1,\dots,e_k)\) with such that \(x\in e_1\setminus e_2\) and \(y\in e_k\setminus e_{k-1}\) whose marked vertices are precisely \(x\) and \(y\).
An \emph{\(x\)-necklace} is a cycle \(C\) such that \(x\) is an inner vertex of \(C\) and the unique marked vertex in \(C\) is \(x\).
If the distinguished vertices are irrelevant, we will refer to these objects simply as nunchakus and necklaces.
See~\cref{fig:nunchaku} for an example of an \(xy\)-nunchaku and an \(x\)-necklace.
Nunchakus and necklaces characterize Maker's win in the game on \(H\):

\begin{figure}[t]
	\centering
	\includegraphics[scale=1.4]{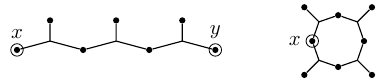}
	\caption{An \(xy\)-nunchaku (left) and an \(x\)-necklace (right).
	Marked vertices are encircled.}
	\label{fig:nunchaku}
\end{figure}

\begin{theoremq}[{\cite[Theorem~3.23, simplified]{galliot2025makerbreakersolvedpolynomialtime}}]\label{thm:3_rounds_chara}
	Let \(H\) be a 3-uniform hypergraph.
	Then \(H\) is a Maker win if and only if Maker has a strategy ensuring that there is a nunchaku or a necklace in the updated marked hypergraph obtained after at most three full rounds of play on \(H\).
\end{theoremq}

\section{Analysis of Claw-Connected Graphs with Maximum Degree at least 4}\label{sec:analysis}
In this section, we analyze the unbiased claw game played on a graph \(G\) with maximum degree at least~4.
Each connected component of \(\incG{G}\) can be seen as an independent game.
Therefore, we will primarily restrict ourselves to claw-connected graphs.
In this setting, Maker's win is fully characterized by the results of this section.
We also give an algorithm with running time \(\Oh{|V(G)|}\) for determining the winner of the unbiased claw game played on a claw-connected graph with maximum degree different from~3 in \cref{sec:alg}, see \cref{thm:claw_con_alg}.
Claw-connected graphs with maximum degree exactly~3 are left as an open problem for future work.
If such graphs also admit a linear-time checkable winning condition, the unbiased claw game can be solved in linear time, see also \cref{cor:decomp_alg}.

We begin with two simple sufficient conditions for Maker's win.
The first condition, \cref{prop:deg5_trivial}, arises from the trivial fact that a vertex with degree at least~5 allows Maker to win because Maker moves first.
Although trivial, this condition is important as it limits the class of interesting graphs to those graphs with maximum degree~4.
Thus, the number of edges can be assumed to be linearly bounded by the number of vertices.
The second condition, \cref{lem:two_deg_4}, provides additional structural insight:
A claw-connected graph containing at least two vertices with degree~4 is Maker's win.

\begin{propositionq}\label{prop:deg5_trivial}
	Let \(G\) be a graph with \(\Delta(G)\geq5\).
	Then Maker wins on \(\incG{G}\).
\end{propositionq}

\begin{lemma}\label{lem:two_deg_4}
	Let \(G\) be a claw-connected graph with \(\Delta(G)=4\) containing at least two vertices of degree exactly~4.
	Then Maker wins on \(\incG{G}\).
\end{lemma}
\begin{proof}
	Let \(v\) and \(w\) be different vertices of degree exactly~4.
	For \(u\in V(G)\), let \(E_u\) be the set of edges incident to \(u\).
	Let \(P=e_1\dots e_k\) be a shortest path from \(E_v\) to \(E_w\) in \(\incG{G}\).
	We set \(P=e_1\) for an edge \({e_1\in E_v\cap E_w}\) in the case that \(E_v\cap E_w\neq\emptyset\).
	Without loss of generality, let \(v\) and \(w\) be chosen such that the resulting path \(P\) is shortest among all possible pairs of different degree~4 vertices.
	
	Maker first claims \(e_1\in E_v\).
	Breaker must then respond by claiming an edge \(a_1\) from \(E_v\) as otherwise Maker wins by claiming available edges from \(E_v\) in the next two rounds.
	There is at most one edge incident to both \(v\) and \(w\), so \(E_v\cap E_w\subseteq\{e_1\}\).
	If \(P=e_1\), then Maker wins by claiming available edges from \(E_w\) in the next two rounds.
	Otherwise, Maker claims \(e_k\in E_w\).
	Breaker must respond by claiming an edge \(a_2\) from \(E_w\) as otherwise Maker wins by claiming available edges from \(E_w\) in the next two rounds.
	
	None of Breaker's forced responses lie on the path \(P\) as otherwise \(P\) would not be a shortest path between \(E_v\) and \(E_w\) in \(\incG{G}\).
	Additionally, all the internal edge-vertices of \(P\) are not incident to vertices of degree~4 as this contradicts the choice of \(v\) and \(w\).
	Consequently, the path \(P\) is linear.
	In the terminology of~\cite{galliot2025makerbreakersolvedpolynomialtime}, the path \(P\) is an \(e_1e_k\)-nunchaku.
	Thus, Maker wins the game by~\cref{thm:3_rounds_chara}.
\end{proof}

\Cref{lem:two_deg_4} allows us to restrict ourselves to claw-connected graphs containing exactly one vertex with degree~4.
The characterization of Maker's win in this setting is the main result of this section, \cref{thm:degree_4_chara}.
For the proof of \cref{thm:degree_4_chara}, we require the following two results:

\begin{lemmaq}[{\cite[Corollary~4.16]{DBLP:journals/corr/abs-2607-18569}}]\label{lem:preprocessing}
	Let \(H=A\cup B\) be the union of two hypergraphs \(A\) and \(B\) such that \({V(A)\cap V(B)=\{v\}}\) and Breaker wins on \(B^{+v}\).
	Then Maker wins on \(H\) if and only if Maker wins on \(A\).
\end{lemmaq}

\begin{lemmaq}[Reformulation of {\cite[Lemma~2.4]{DBLP:journals/dam/DucheneGINOPS25}}]\label{lem:pairing}
	Let \(H\) be a hypergraph in which all hyperedges are maximal and let \(v_1,v_2\in V(H)\) be such that the two sets \(S_i=\{e\setminus \{v_1,v_2\}: v_i\in e\in E(H)\}\) for \(i\in\{1,2\}\) are identical.
	Then Maker wins on \(H\) if and only if Maker wins on \(H^{+v_1}-v_2\).
\end{lemmaq}

The maximality condition is not stated in \cite{DBLP:journals/dam/DucheneGINOPS25} as the authors consider the \(H\)-game, where this condition is always satisfied.
With this condition, the proof from \cite{DBLP:journals/dam/DucheneGINOPS25} also applies to general hypergraphs.

\begin{theoremq}\label{thm:degree_4_chara}
	Let \(G\) be a claw-connected graph with \(\Delta(G)=4\) containing exactly one vertex \(u\) of degree~4.
	Set \(E^*=\{uv:v\in N(u)\}\).
	For \(e\in E^*\), let \(K_e\) be the connected component of \(e\) in \(\incG{G}-(E^*\setminus\{e\})\).
	Then Maker wins on \(\incG{G}\) if and only if one of the following conditions is satisfied:
	\begin{enumerate}[label=(\arabic*)]
		\item\label{item:claw-chara-3-in-comp} Three elements of \(E^*\) are in the same connected component of \(\incG{G}-\claws{E^*}\).
		\item\label{item:claw-chara-biased-win} There exists an edge \(e\in E^*\) such that Maker wins on \(\incG{G}[K_e]^{+e}\).
		\item\label{item:claw-chara-good-anchor} There exists a connected component \(K\) of \(\incG{G}-\claws{E^*}\) containing exactly two edges \(e_1,e_2\) of \(E^*\) such that Maker wins on \(\incG{G}+\{e_1,e_2\}\).\qedhere
	\end{enumerate}
\end{theoremq}

The formulation of \cref{thm:degree_4_chara} is slightly redundant:
Conditions \labelcref{item:claw-chara-3-in-comp,item:claw-chara-good-anchor} could be combined into a single condition that allows adding a new hyperedge whenever two elements of \(E^*\) are in the same connected component of \(\incG{G}-\claws{E^*}\).
While such a formulation is more succinct, this redundant formulation is more useful for obtaining an efficient algorithm, which is our goal.

\begin{proof}[Proof of~\cref{thm:degree_4_chara}]
	Set \(V\coloneqq V(G)\) and \(E\coloneqq E(G)\).
	Every edge \(e\in E\) is contained in some \(K_{e'}\) for an edge \(e'\in E^*\):
	Let \(P\) be a shortest path from \(e\) to \(E^*\) in \(\incG{G}\).
	Such a path exists as \(G\) is claw-connected.
	Let \(e'\) be the last edge-vertex of \(P\).
	Then \(e\) lies in \(K_{e'}\) because \(P\) cannot use a hyperedge from \(\claws{E^*}\) as otherwise it would contain another element of \(E^*\) contradicting the fact that it is a shortest path.
	
	\textbf{Sufficiency}.
	First, we will show that conditions~\labelcref{item:claw-chara-3-in-comp} to~\labelcref{item:claw-chara-good-anchor} are sufficient for Maker's win.
	The proof commonly uses the fact that if Maker claims an edge from \(E^*\) when neither Maker nor Breaker has claimed an edge from \(E^*\) before, Breaker must also claim an edge from \(E ^*\) to avoid losing.
	If Maker claims an edge from \(E^*\), but Breaker does not reply with an element of \(E^*\), Maker wins by claiming any available edges from \(E^*\) in the next two rounds.
	
	\textbf{Condition \labelcref{item:claw-chara-3-in-comp}}.
	Let \(e_0,e_1,e_2\in E^*\) be pairwise distinct and in the same connected component \(K\) of \(\incG{G}-\claws{E^*}\) such that the shortest paths \(P_1,P_2\) from \(e_0\) to \(e_1\) and from \(e_0\) to \(e_2\) minimize \({|V(P_1)|+|V(P_2)|}\) among all possible choices of such paths.
	Note that \(E(P_i)\cap E^*=\{e_0,e_i\}\) for all \(i\in\{1,2\}\), i.e.~the only edge-vertices from \(E^*\) used by hyperedges in \(E(P_i)\) are \(e_0\) and \(e_i\).
	No hyperedge in \(P_i\) except for the last one can contain elements of \(E^*\setminus\{e_0\}\) as otherwise we can obtain paths \(Q_i\) from \(e_0\) to \(e'_i\in E^*\), \(i\in\{1,2\}\), such that \({|Q_1|+|Q_2|<|P_1|+|P_2|}\), a contradiction.
	The only element of \(E^*\) contained in the last hyperedge of \(P_i\) is \(e_i\) as otherwise this hyperedge would be an element of \(\claws{E^*}\) and these hyperedges are not present in \(\incG{G}-\claws{E^*}\).
	Thus, the paths \(P_1\) and \(P_2\) are linear.
	
	Maker first picks \(e_0\).
	Breaker must respond by picking some edge \({a_1\in E^*}\) as discussed at the beginning of the proof.
	Regardless of the choice of \(a_1\), \({\incG{G[E^*]}-a_1}\) consists of a single hyperedge \(h\).
	At least one of \(e_1\) and \(e_2\), say \(e_1\), is not claimed by Breaker and is thus an element of \(h\).
	As the hyperedges in \(E(P_1)\) only intersect \(E^*\) in \(\{e_0,e_1\}\), the path \(P_1\) is unaffected by Breaker's choice of \(a_1\).
	Combining \(P_1\) with the hyperedge \(h\) yields a cycle \(C\) in \(\incG{G}-a_1\).
	In the terminology of~\cite{galliot2025makerbreakersolvedpolynomialtime}, the cycle \(C\) is an \(e_0\)-necklace.
	Thus, Maker wins the game by~\cref{thm:3_rounds_chara}.
	
	\textbf{Condition \labelcref{item:claw-chara-biased-win}}.
	Let \(e\in E^*\) be such that Maker wins on \(\incG{G}[K_e]^{+e}\).
	Maker wins on \(\incG{G}\) by first picking \(e\).
	Breaker must reply with some \(a\in E^*\) as discussed at the beginning of the proof.
	Then Maker wins by following a winning strategy for the game on \(\incG{G}[K_e]^{+e}\).
	This is possible because \(a\in E^*\setminus\{e\}\), so Breaker's move has no impact on the game on \(\incG{G}[K_e]^{+e}\).
	
	\textbf{Condition \labelcref{item:claw-chara-good-anchor}}.
	Let \(E^*\setminus\{e_1,e_2\}=\{e_3,e_4\}\) and set \(H\coloneqq\incG{G}+\{e_1,e_2\}\).
	We may assume that neither condition~\labelcref{item:claw-chara-3-in-comp} nor condition~\labelcref{item:claw-chara-biased-win} is satisfied.
	As condition~\labelcref{item:claw-chara-3-in-comp} is not satisfied, every connected component of \(\incG{G}-\claws{E^*}\) contains at most two elements of \(E^*\).
	For \(i\in\{1,2\}\), set \({R_i=K_{e_i}\setminus K_{e_{3-i}}}\).
	Then \(\incG{G}\) decomposes into \(H_1=\incG{G}-\left(R_i\setminus\{e_i\}\right)\) and \(H_2=\incG{G}[R_i]\) with \({V(H_1)\cap V(H_2)=\{e_i\}}\).
	As condition~\labelcref{item:claw-chara-biased-win} is not satisfied, Breaker wins on \(\incG{G}[K_{e_i}]^{+e_i}\), so Breaker also wins on \(H_2^{+e_i}\).
	By \cref{lem:preprocessing}, Maker wins on \(\incG{G}\) if and only if Maker wins on \(H_1\).
	Similarly, \(H\) decomposes into \(H'_1=H[V(H_1)]\) and \(H'_2=H[V(H_2)]\) with \({V(H'_1)\cap V(H'_2)=\{e_i\}}\).
	Observe that \(H'_1\cong H_1+\{e_1,e_2\}\) and \(H'_2\cong H_2\).
	By \cref{lem:preprocessing}, Maker wins on \(H\) if and only if Maker wins on \(H'_1\).
	We may therefore assume that \({\incG{G}=H_1}\), i.e.~\({K_{e_1}\Delta K_{e_2}=\{e_1,e_2\}}\), where \(K_{e_1}\Delta K_{e_2}\) denotes the symmetric difference.
	By an analogous argument, we may assume that \(K_{e_3}\Delta K_{e_4}=\{e_3,e_4\}\).
	Assume for contradiction that Breaker wins on \(\incG{G}\).
	We will show that Breaker then also wins on \(H\).
	We distinguish two cases.
	
	\textbf{Case A}.
	Assume that exactly one connected component of \(\incG{G}-\claws{E^*}\) contains two elements of \(E^*\).
	Then \(K_{e_3}\cap K_{e_4}=\emptyset\), so \(K_e=\{e\}\) for \(e\in\{e_3,e_4\}\).
	Thus, \(e_3\) and \(e_4\) are only involved in the four hyperedges \(\{e_1,e_2,e_3\}\), \(\{e_1,e_2,e_4\}\), \(\{e_1,e_3,e_4\}\) and \(\{e_2,e_3,e_4\}\).
	Set \(S_i=\{h\setminus\{e_3,e_4\}:e_i\in h\in E(\incG{G})\}\) for \(i\in\{3,4\}\).
	Then \(S_3=S_4\).
	By~\cref{lem:pairing}, Maker wins on \(\incG{G}\) if and only if Maker wins on \(H'\coloneqq\incG{G}^{+e_3}-e_4\).
	Thus, by our assumption, Breaker wins on \(H'\).
	Note that deleting the vertices \(e_3\) and \(e_4\) from \(H\) results in \(H'\).
	Then Breaker wins on \(H\) by pairing \(e_3\) and \(e_4\) in \(H\): 
	Breaker claims \(e_3\) if Maker claims \(e_4\) and vice versa.
	The hyperedges \(\{e_1,e_2,e_3\}\) and \(\{e_1,e_2,e_4\}\) may be removed from \(H\) as they are supersets of the hyperedge \(\{e_1,e_2\}\).
	Thus, \(e_3\) and \(e_4\) only appear in \(\{e_1,e_3,e_4\}\) and \(\{e_2,e_3,e_4\}\) which are both made irrelevant by pairing \(e_3\) and \(e_4\).
	Consequently, Breaker wins on \(H\), contradicting condition~\labelcref{item:claw-chara-good-anchor}.
	
	\textbf{Case B}.
	In the remaining case, the two connected components of \(\incG{G}-\claws{E^*}\) each contain exactly two elements of \(E^*\).
	Set \(K_1\coloneqq K_{e_1}\cup K_{e_2}\) and \(K_2\coloneqq K_{e_3}\cup K_{e_4}\).
	The sets \(K_1\) and \(K_2\) are a partition of \(E\).
	We say that \(e_1\) and \(e_2\) are conjugates of each other.
	By our assumption, Breaker wins on \(\incG{G}\), so let \(B\) be a winning strategy for Breaker on \(\incG{G}\).
	Let \(M\) be a winning strategy for Maker on \(H\) playing in such a way that, in any situation, its next move minimizes the maximum number of rounds required to achieve a nunchaku or necklace where the maximum is taken over all Breaker strategies.
	By~\cref{thm:3_rounds_chara}, \(M\) requires at most~3 rounds to claim a nunchaku or a necklace.
	As \(\incG{G}\) and \(H\) have the same vertices, we can also view \(B\) as a strategy on \(H\).
	Consider the first three rounds of the game of \(M\) against \(B\) on \(H\).
	
	In the following, let \(m_1,m_2,m_3\) be the edges claimed by \(M\) and \(a_1,a_2,a_3\) be the edges claimed by \(B\) in the first three rounds.
	Maker must claim an \(xy\)-nunchaku or an \(x\)-necklace in \(\incG{G}-\{a_1,a_2,a_3\}\), where \(x,y\in\{m_1,m_2,m_3\}\).
	As \(B\) is a winning strategy for the game on \(\incG{G}\), the nunchaku or the necklace must involve the new hyperedge \(\{e_1,e_2\}\).
	In particular, \(B\) claims neither \(e_1\) nor \(e_2\).
	If an edge-vertex \(e\in \{e_3,e_4\}\) is not claimed by Breaker, any nunchaku or necklace \(S\) in \(H\) yields a nunchaku or necklace in \(\incG{G}\) by either replacing the hyperedge \(\{e_1,e_2\}\) in \(S\) with \(\{e_1,e_2,e\}\) if \(e\) does not occur in \(S\) or by replacing the subpath from \(e_1\) to \(e\) over \(e_2\) in \(S\) with \(\{e_1,e_2,e\}\) contradicting the fact that \(B\) is a winning strategy on \(\incG{G}\).
	Thus, \(B\) claims both \(e_3\) and \(e_4\).
	Without loss of generality, let the indices be such that \(e_3\) is claimed before \(e_4\).
	We now give a brief overview of the remaining proof strategy.
	
	\begin{figure}[t]
		\centering
		\hspace{2mm}
		$\vcenter{\hbox{\includegraphics[scale=0.9]{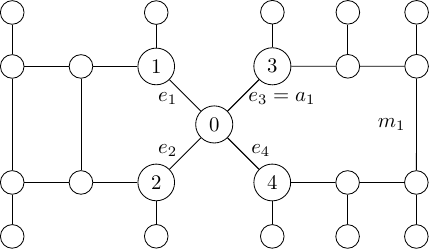}}}$
		\hspace{5mm}
		$\vcenter{\hbox{\includegraphics[scale=0.9]{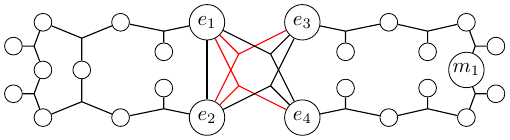}}}$
		\hspace{2mm}
		\caption{A graph \(G\) (left) satisfying condition~\labelcref{item:claw-chara-good-anchor} of~\cref{thm:degree_4_chara} and the hypergraph \(H=\incG{G}+\{e_1,e_2\}\) (right).
			Hyperedges are drawn in a claw-like way as in~\cite{kutz2004thesis,galliot2025makerbreakersolvedpolynomialtime}.
			The hyperedges which are a superset of \(\{e_1,e_2\}\) are colored red.
			Potential choices for the edges \(m_1,m_2,m_3\) and \(a_1,a_2\) in the game of \(M\) against \(B\) on \(H\) are also given.}
		\label{fig:round_overview}
	\end{figure}
	
	The graph from \cref{fig:round_overview} (left) shows a possible choice for Maker in the first round.
	The corresponding hypergraph is displayed on the right.
	The vertex \(0\) is the unique vertex with degree~4.
	The left half of the graph consists of the edges in \(K_1\) and the right half consists of the edges in \(K_2\).
	In the analysis of the first round, we show that \(m_1\in K_2\setminus E^*\) and \(a_1=e_3\).
	This splits the hypergraph into two components as the remaining element of \(\claws{E^*}\) is a superset of \(\{e_1,e_2\}\).
	We then use~\cref{lem:pairing} to obtain a contradiction.
	
	\textbf{First round}.
	Recall that \(K_1\) and \(K_2\) are a partition of \(E\).
	Assume for contradiction that \(m_1\in E^*\).
	As discussed above, \(B\) claims both \(e_3\) and \(e_4\) in the first three rounds, so \(m_2\in\{e_1,e_2\}\).
	Breaker must reply with an edge in \(E^*\) as discussed at the beginning of the proof.
	In fact, \(B\) must claim the conjugate of \(m_1\) to prevent Maker from claiming an \(m_1\)-necklace because \(e_1\) and \(e_2\) lie on a common cycle in \(\incG{G}\) as they are both elements of \(K_1\).
	This contradicts the fact that \(B\) claims neither \(e_1\) nor \(e_2\), so \(m_1\notin E^*\).
	
	Assume for contradiction that \(m_1\notin K_2\), i.e.~\(m_1\in K_1\setminus E^*\).
	If \(a_1\) separates \(m_1\) from \(\{e_1,e_2\}\), then Breaker wins on \(H[\concomp{H-a_1}{m_1}]^{+m_1}\) by playing according to \(B\) as this is a subhypergraph of \(\incG{G}^{+m_1}-a_1\), which does not contain the new hyperedge \(\{e_1,e_2\}\).
	In particular, \(M\) cannot claim a nunchaku or necklace in \(H[\concomp{H-a_1}{m_1}]^{+m_1}\).
	Thus, \(M\) must claim one of these structures in \(H-\concomp{H-a_1}{m_1}\) in the remaining two rounds.
	Then \(M\) could have claimed a winning structure in fewer rounds by playing only in \(H-\concomp{H-a_1}{m_1}\).
	This contradicts that \(M\) minimizes the rounds required to claim a nunchaku or necklace.
	Thus, \(a_1\) does not separate \(m_1\) from \(\{e_1,e_2\}\).
	We distinguish the cases \(a_1\neq e_3\) and \(a_1=e_3\).
	
	Assume first that \(a_1\neq e_3\).
	Then \(a_1\notin E^*\) as \(B\) claims neither \(e_1\) nor \(e_2\) and \(e_3\) is claimed before \(e_4\).
	Maker then wins against \(B\) on \(\incG{G}\) by claiming \(e_1\) in the next round.
	As discussed at the beginning of the proof, \(B\) must reply with an element of \(E^*\).
	As \(m_1\in K_1\), \(B\) cannot separate \(e_1\) from \(m_1\) in \(\incG{G}-a_1\) with an element of \(E^*\) resulting in an \(m_1e_1\)-nunchaku, contradicting the fact that \(B\) is a winning strategy on \(\incG{G}\).
	
	Next, assume \(a_1=e_3\).
	Then \(m_1\) does not lie on a cycle in \(\incG{G}\) as otherwise \(B\) would have allowed Maker to claim an \(m_1\)-necklace.
	Therefore, \(m_1\) does not lie on a cycle in \(H\) either.
	By claiming the first available edge-vertex \(a'_1\) on a shortest path from \(m_1\) to \(\{e_1,e_2\}\) in \(H\) instead of \(e_3\), Breaker separates \(m_1\) from \(\{e_1,e_2\}\).
	As before, \(M\) must then claim a winning structure in \(H-\concomp{H-a'_1}{m_1}\) within the remaining two rounds, so \(M\) could have skipped its first move. 
	Again, this contradicts the fact that \(M\) minimizes the number of rounds required to claim a nunchaku or necklace.
	We conclude that \(m_1\in K_2\setminus E^*\).
	
	Assume for contradiction that \(a_1\neq e_3\).
	Similar to the above analysis, \(a_1\) must separate \(\{e_3,e_4\}\) from \(m_1\).
	But this again implies that \(M\) is able to claim a nunchaku or necklace even if it skipped its first turn by an analogous argument.
	This contradicts the choice of \(M\).
	Thus, \(a_1\in\{e_3,e_4\}\), so \(a_1=e_3\) by the assumption on the indices.
	
	\textbf{Obtaining the contradiction}.
	After the first round, the current state of the game is captured by the hypergraph \(H'=H^{+m_1}-a_1\).
	As an example, \cref{fig:round_two_overview} shows the updated hypergraph resulting from the hypergraph in \cref{fig:round_overview}.
	As \(a_1\) was removed from \(H\), the only element of \(\claws{E^*}\) in \(H'\) is \(\{e_1,e_2,e_4\}\).
	However, this is a superset of \(\{e_1,e_2\}\), so Maker wins on \(H'\) if and only if Maker wins on the hypergraph \(H''\) obtained from \(H'\) by removing \(\{e_1,e_2,e_4\}\).
	Because \(M\) is a winning strategy for Maker on \(H\), Maker wins on \(H'\), so Maker also wins on \(H''\).
	The hypergraph \(H''\) decomposes into \(H_1=H''[K_1]\) and \(H_2=H''-K_1\).
	Note that \(H_1\cong H[K_1]\) and \(H_2\subseteq\incG{G}^{+m_1}-a_1\).
	As \(B\) is a winning strategy on \(\incG{G}\), Breaker wins on \(\incG{G}^{+m_1}-a_1\) by continuing to play according to \(B\).
	Thus, Breaker wins on \(H_2\).
	As Maker wins on \(H''\), Maker must win on \(H_1\).
	
	\begin{figure}[t]
		\centering
		\includegraphics{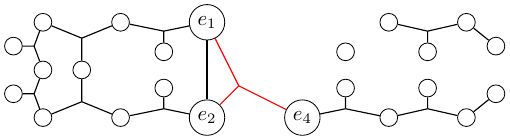}
		\caption{The updated hypergraph \(H^{+m_1}-a_1\) after the first round.}
		\label{fig:round_two_overview}
	\end{figure}
	
	Let \(H_3=\incG{G}[K_1\cup\{e_3,e_4\}]\).
	In \(H_3\), \(e_3\) and \(e_4\) only occur in the four hyperedges from \(\claws{E^*}\).
	Set \(S_i=\{h\setminus\{e_3,e_4\}:e_i\in h\in E(\incG{G})\}\) for \(i\in\{3,4\}\).
	Then \(S_3\) and \(S_4\) are identical.
	By~\cref{lem:pairing}, Maker wins on \(H_3\) if and only if Maker wins on \(H_3^{+e_3}-e_4\cong H_1\).
	Thus, Maker wins on \(H_3\).
	Then Maker also wins on \(\incG{G}\) as \(H_3\subseteq\incG{G}\), contradicting the assumption that Breaker win on \(\incG{G}\).
	
	\textbf{Necessity}.
	Lastly, we show that one of the conditions~\labelcref{item:claw-chara-3-in-comp} to~\labelcref{item:claw-chara-good-anchor} is necessary for Maker's win.
	If there exists a connected component of \(\incG{G}-\claws{E^*}\) containing exactly two elements \(e_1\) and \(e_2\) of \(E^*\), Breaker wins on \(\incG{G}+\{e_1,e_2\}\) as condition~\labelcref{item:claw-chara-good-anchor} is not satisfied.
	In particular, Breaker also wins on \(\incG{G}\).
	Otherwise, every connected component of \(\incG{G}-\claws{E^*}\) contains exactly one element of \(E^*\) as condition~\labelcref{item:claw-chara-3-in-comp} is not satisfied.
	Then \(\incG{G}\) decomposes into \(H_1=\incG{G}-\left(K_e\setminus\{e\}\right)\) and \(H_2=\incG{G[K_e]}\) with \({V(H_1)\cap V(H_2)=\{e\}}\) for any \(e\in E^*\).
	Breaker wins on \(\incG{G[K_e]}^{+e}\) for all \(e\in E^*\) as condition~\labelcref{item:claw-chara-biased-win} is not satisfied.
	By~\cref{lem:preprocessing}, Maker wins on \(\incG{G}\) if and only if Maker wins on \(H_1\).
	Consequently, we may assume that \(K_e=\{e\}\) for any \(e\in E^*\), i.e.~\(G\iso K_{1,4}\).
	Breaker clearly wins on \(\incG{K_{1,4}}\) because this hypergraph has only four edge-vertices.
\end{proof}

This completes our characterization of the unbiased claw game on claw-connected graphs \(G\) with \(\Delta(G)\geq4\).
If \(\Delta(G)\geq5\), then Maker wins by \cref{prop:deg5_trivial}.
Otherwise, \(\Delta(G)=4\).
If \(G\) contains two vertices of degree~4, Maker wins by \cref{lem:two_deg_4}.
This leaves only claw-connected graphs containing exactly one vertex with degree~4, which are covered by \cref{thm:degree_4_chara}.

\subsection{Reducing to Linear Hypergraphs}
In this section, we prove the final results required for our algorithmic results, \cref{cor:degree4_edge_case_1,cor:degree4_edge_case_2}. 
In the unbiased claw game, the presence of vertices with degree at least~4 enforces that \(\incG{G}\) is non-linear.
This prevents the use of efficient algorithms for the Maker-Breaker game on linear hypergraphs with rank at most~3, see~\cite{kutz2004thesis,DBLP:journals/corr/abs-2607-18569}, as subroutines.
Instead, slower algorithms for the Maker-Breaker game on general hypergraphs with rank at most~3 have to be used, see~\cite{galliot2025makerbreakersolvedpolynomialtime}.

With \cref{cor:degree4_edge_case_1,cor:degree4_edge_case_2}, we reduce the outcome of the unbiased claw game to the outcome of a constant number of Maker-Breaker games on linear hypergraphs with rank at most~3 and size bounded by the original game.
The size of these hypergraphs is bounded by the size of the graph.
\Cref{cor:degree4_edge_case_1} addresses the case where exactly one connected component of \(\incG{G}-\claws{E^*}\) contains exactly two elements of \(E^*\).
\Cref{cor:degree4_edge_case_2} covers the case where two connected components of \(\incG{G}-\claws{E^*}\) contain exactly two elements of \(E^*\).

\begin{corollary}\label{cor:degree4_edge_case_1}
	Let \(G\) be a claw-connected graph with \(\Delta(G)=4\) containing exactly one vertex \(u\) of degree~4.
	Set \(E^*=\{uv:v\in N(u)\}\).
	For \(e\in E^*\), let \(K_e\) be the connected component of \(e\) in \({\incG{G}-(E^*\setminus\{e\})}\).
	Assume further that Breaker wins on \(\incG{G}[K_e]^{+e}\) for every \(e\in E^*\) and exactly one connected component of \({\incG{G}-\claws{E^*}}\) contains exactly two elements of \(E^*\).
	Let \(e_1\) and \(e_2\) be those two elements.
	Then Maker wins on \(\incG{G}\) if and only if Maker wins on \(\incG{G}[K_{e_1}\cup K_{e_2}]+\{e_1,e_2\}\).
\end{corollary}
\begin{proof}
	For \(e\in\{e_3,e_4\}\), \(\incG{G}\) decomposes into \(H_1=\incG{G}-(K_e\setminus\{e\})\) and \(H_2=\incG{G}[K_e]\) with \({V(H_1)\cap V(H_2)=\{e\}}\).
	Breaker wins on \(\incG{G}[K_e]^{+e}\) by assumption.
	By \cref{lem:preprocessing}, Maker wins on \(\incG{G}\) if and only if Maker wins on \(H_1\), so we may assume that \(K_e=\{e\}\).
	By \cref{thm:degree_4_chara}, Maker wins on \(\incG{G}\) if and only if Maker wins on \(\incG{G}+\{e_1,e_2\}\).
	Observe that any hyperedge in \(\incG{G}+\{e_1,e_2\}\) that is a strict superset of \(\{e_1,e_2\}\) can be removed without changing the outcome of the game.
	In the resulting hypergraph, \(e_3\) and \(e_4\) are isolated vertices, so they may be removed as well.
	This results in the hypergraph \(\incG{G}[K_{e_1}\cup K_{e_2}]+\{e_1,e_2\}\).
\end{proof}

\begin{corollary}\label{cor:degree4_edge_case_2}
	Let \(G\) be a claw-connected graph with \(\Delta(G)=4\) containing exactly one vertex \(u\) of degree~4.
	Set \(E^*=\{e_1,e_2,e_3,e_4\}=\{uv:v\in N(u)\}\).
	Assume further that every connected component of \({H=\incG{G}-\claws{E^*}}\) contains exactly two elements of \(E^*\).
	Let \(\concomp{H}{e_1}=\concomp{H}{e_2}\) and \(\concomp{H}{e_3}=\concomp{H}{e_4}\).
	Then Maker wins on \(\incG{G}\) if and only if Maker wins on \(\incG{G}+\{\{e_1,e_2\},\{e_3,e_4\}\}\).
\end{corollary}
\begin{proof}
	As additional hyperedges only provide Maker with more opportunities to win, Maker wins on \(\incG{G}+\{\{e_1,e_2\},\{e_3,e_4\}\}\) if Maker wins on \(\incG{G}\).
	Assume Maker wins on \(\incG{G}+\{\{e_1,e_2\},\{e_3,e_4\}\}\).
	Observe that any hyperedge in \(\incG{G}+\{\{e_1,e_2\},\{e_3,e_4\}\}\) that is a strict superset of \(\{e_1,e_2\}\) or \(\{e_3,e_4\}\) can be removed without changing the outcome of the game.
	Let \(H\) be the resulting hypergraph.
	This hypergraph decomposes into \(H_1=H[\concomp{H}{e_1}]\) and \(H_2=H[\concomp{H}{e_3}]\).
	If Breaker wins on both \(H_1\) and \(H_2\), Breaker also wins on \(H\) by playing separately in each connected component.
	As Maker wins on \(H\), one of \(H_1\) and \(H_2\) must be Maker's win, say \(H_1\).
	Because \({H_1\subseteq\incG{G}+\{e_1,e_2\}}\), Maker wins on \(\incG{G}\) by condition~\labelcref{item:claw-chara-good-anchor} of \cref{thm:degree_4_chara}.
\end{proof}

\section{Algorithmic Results}\label{sec:alg}
In this section, we use the results of \cref{sec:analysis} to design an algorithm with running time \(\Oh{|V(G)|}\) for the unbiased claw game played on a claw-connected graph \(G\) with maximum degree different from~3, see \cref{thm:claw_con_alg}.
We then use this algorithm to construct a linear-time decomposition algorithm which either outputs the winner of the unbiased claw game or decomposes the graph into edge-disjoint claw-connected subgraphs with maximum degree exactly~3 such that the winner of the game can be determined from the winners of the games on the subgraphs, see \cref{cor:decomp_alg}.
This allows settling the time complexity of the unbiased claw game in future work by analyzing only claw-connected graphs with maximum degree~3.
As an application of this decomposition, we obtain an \(\Oh{|V(G)|^3}\)-algorithm for the unbiased claw game on general graphs.
We require the following lemma:

\begin{figure}[t]
	\begin{enumerate}[label=\arabic*, widest=11, align=right]
		\item[]\textbf{Input}: A claw-connected graph \(G\) with \(\Delta(G)\neq3\).
		\item[]\textbf{Output}: The winner of the Maker-Breaker-game on \(\incG{G}\).
		\item\label{dn4:deg5} Iterate over \(E(G)\) to determine the degree of every vertex.
		During the iteration, if any vertex is found to have degree at least~5 due to the edges enumerated so far, abort the loop and output ``Maker''.
		\item Iterate over the vertices to determine \(\Delta(G)\) and the number of vertices with degree~4.
		\item\label{dn4:deg2} If \(\Delta(G)\leq2\), output ``Breaker``.
		\item\label{dn4:two_deg4} If there are at least two vertices with degree~4, output ``Maker''.
		\item Determine the set \(E^*\) of edges incident to the unique degree~4 vertex.
		\item\label{dn4:hypergraph1} Initialize \(\incG{G}\coloneqq(E(G),\emptyset)\).
		\item\label{dn4:hypergraph2} Iterate over every vertex \(v\) and add a hyperedge for every claw using only edges incident to \(v\).
		\item Compute the connected components \(K_1,\dots,K_\ell\) of \(\incG{G}-\claws{E^*}\) using breadth-first search.
		\item\label{dn4:3_in_conp} If any of these components contains at least~3 elements of \(E^*\), output ``Maker''.
		\item\label{dn4:cond_2} For every edge \(e\in E^*\), compute \(K_e\) as defined in \cref{thm:degree_4_chara} using breadth-first search starting from \(e\) in \(\incG{G}-\left(E^*\setminus\{e\}\right)\) and set \(H_e\coloneqq\incG{G}[K_e]^{+e}\).
		Use~\cref{lem:alg-non-3-uniform} to determine the winner of the game on \(H_e\).
		If Maker wins on \(H_e\), output ``Maker''.
		\item\label{dn4:simple} If every connected component \(K_1,\dots,K_\ell\) contains only one element of \(E^*\), output ``Breaker''.
		\item\label{dn4:edgecase_1} If exactly one connected component of \(K_1,\dots,K_\ell\) contains two elements \(e_1\) and \(e_2\) of \(E^*\), compute the hypergraph \(H=(H_{e_1}\cup H_{e_2})+\{e_1,e_2\}\).
		Use~\cref{lem:alg-non-3-uniform} to determine the winner of the game on \(H\) and output it.
		\item\label{dn4:edgecase_2} Let \(E^*=\{e_1,e_2,e_3,e_4\}\) with \(K_{e_1}\cap K_{e_2}\neq\emptyset\) and \(K_{e_3}\cap K_{e_4}\neq\emptyset\).
		Compute \(H_1=(H_{e_1}\cup H_{e_2})+\{e_1,e_2\}\) and \(H_2=(H_{e_3}\cup H_{e_4})+\{e_3,e_4\}\).
		Use~\cref{lem:alg-non-3-uniform} to determine the winner of the game on \(H_1\) and \(H_2\) and output ``Maker'' if Maker wins at least one of the games.
		Otherwise, output ``Breaker''.
	\end{enumerate}
	\renewcommand\figurename{Algorithm}
	\setcounter{enumi}{\value{figure}}
	\setcounter{figure}{0}
	\caption{A linear-time algorithm for the unbiased claw game.}
	\label[algorithm]{alg:dn4}
	\setcounter{figure}{\value{enumi}}
\end{figure}

\begin{lemmaq}[{\cite[Corollary~4.10]{DBLP:journals/corr/abs-2607-18569}}]\label{lem:alg-non-3-uniform}
	Let \(H\) be a connected linear hypergraph of rank at most~3 that is not 3-uniform.
	Then the winner of the Maker-Breaker game on \(H\) can be computed in time \(\Oh{|V(H)|+|E(H)|}\).
\end{lemmaq}

\begin{theorem}\label{thm:claw_con_alg}
	\Cref{alg:dn4} correctly decides the outcome of the unbiased claw game in time \(\Oh{|V(G)|}\).
\end{theorem}
\begin{proof}
	If \cref{alg:dn4} terminates in step~\labelcref{dn4:deg5}, the output is correct by~\cref{prop:deg5_trivial}.
	After this step, every vertex has degree at most~4.
	If \(\Delta(G)\leq2\), Breaker wins on \(\incG{G}\) as there are no claws in \(G\), so step~\labelcref{dn4:deg2} is correct.
	At this point, it must be the case that \(\Delta(G)=4\) as \(\Delta(G)\neq3\).
	If \(G\) contains at least two vertices of degree~4, Maker wins on \(\incG{G}\) by \cref{lem:two_deg_4}, so step~\labelcref{dn4:two_deg4} is correct.
	This leaves only the case where \(G\) contains exactly one vertex \(u\) of degree~4.
	
	Step~\labelcref{dn4:3_in_conp} is correct by condition~\labelcref{item:claw-chara-3-in-comp} of \cref{thm:degree_4_chara}.
	If the algorithm does not terminate in step~\labelcref{dn4:3_in_conp}, every connected component contains at most~2 elements of \(E^*\).
	Observe that the hypergraph \(H_e\) from step~\labelcref{dn4:cond_2} is linear and every connected component of \(H_e\) contains a hyperedge of size~2, so~\cref{lem:alg-non-3-uniform} is indeed applicable.
	If Maker wins on \(H_e\) for any \(e\in E^*\), Maker wins on \(\incG{G}\) by condition~\labelcref{item:claw-chara-biased-win} of \cref{thm:degree_4_chara}, so this step is correct.
	Otherwise, only condition~\labelcref{item:claw-chara-good-anchor} can lead to Maker's win.

	If every connected component \(K_1,\dots,K_\ell\) contains only one element of \(E^*\), condition~\labelcref{item:claw-chara-good-anchor} is not satisfied and so Breaker wins by~\cref{thm:degree_4_chara}, so step~\labelcref{dn4:simple} is correct.
	If exactly one connected component of \(K_1,\dots,K_\ell\) contains two elements \(e_1\) and \(e_2\) of \(E^*\), Maker wins on \(\incG{G}\) if and only if Maker wins on \(H\) by \cref{cor:degree4_edge_case_1}, so step~\labelcref{dn4:edgecase_1} is correct.
	The hypergraph \(H\) is linear, connected and not 3-regular, so~\cref{lem:alg-non-3-uniform} is applicable.
		
	This leaves the case where the two connected components of \(\incG{G}-\claws{E^*}\) each contain two elements of \(E^*\).
	By \cref{cor:degree4_edge_case_2}, Maker wins on \(\incG{G}\) if and only if Maker wins on \(\incG{G}+\{\{e_1,e_2\},\{e_3,e_4\}\}\).
	Removing the non-maximal hyperedges \(\claws{E^*}\) results in the two connected components \(H_1\) and \(H_2\).
	These hypergraphs are linear and contain a hyperedge of size~2, so step~\labelcref{dn4:edgecase_2} is also correct.
	
	It remains to show that the running time is \(\Oh{|V(G)|}\).
	As the loop in step~\labelcref{dn4:deg5} is aborted once a vertex with degree~5 is found, at most \(2|V(G)|+1\) edges are inspected by the pigeonhole principle.
	After this step, every vertex has degree at most~4, so \(|E(G)|\leq2|V(G)|\).
	This also implies that every vertex lies on at most~4 claws, so \(|\claws{G}|\leq4|V(G)|\) and \(|V(\incG{G})|+|E(\incG{G})|\leq6|V(G)|\).
	Consequently, computing \(\incG{G}\) in steps~\labelcref{dn4:hypergraph1,dn4:hypergraph2} only takes time \(\Oh{|V(G)|}\).
	The remaining steps are all linear in the size of \(\incG{G}\), so the running time follows.
\end{proof}

Each subgraph induced by the edges of a connected component of \(\incG{G}\) is claw-connected by definition.
Applying~\cref{alg:dn4} when possible leaves only claw-connected subgraphs with maximum degree~3.
This leads to \cref{cor:decomp_alg}:
\begin{corollary}\label{cor:decomp_alg}
	There exists an algorithm that, given a graph \(G\), either outputs the winner of the unbiased claw game on \(G\) or computes pairwise disjoint sets \(E_1,\dots,E_\ell\subseteq E(G)\) such that, for all \(1\leq i\leq\ell\), the graph \(G[E_i]\) is claw-connected and \(\Delta(G[E_i])=3\).
	Additionally, Maker wins on \(\incG{G}\) if and only if there exists \(1\leq i\leq\ell\) such that Maker wins on \(\incG{G[E_i]}\).
	This algorithm runs in time \(\Oh{|V(G)|}\).
\end{corollary}
\begin{proof}
	The algorithm proceeds as follows:
	First iterate over \(E(G)\) to determine the degree of every vertex.
	During the loop, if any vertex is found to have degree at least~5 due to the edges enumerated so far, the loop is aborted and the algorithm outputs ``Maker''.
	This is correct by \cref{prop:deg5_trivial}.
	If the loop finishes, every vertex has degree at most~4.
	The algorithm then iterates over every vertex marking the edges incident to vertices of degree~4.

	As \(\Delta(G)\leq 4\), every vertex is the center of at most~4 claws.
	Compute \(\incG{G}\) by starting from the hypergraph \((E(G),\emptyset)\).
	Iterate over every vertex \(v\) and add a hyperedge for every claw using only edges incident to \(v\).
	Compute the connected components \(K_1,\dots,K_\ell\) of \(\incG{G}\) using breadth-first search.
	
	Initialize \(S=\emptyset\).
	Iterate over the indices \(1,\dots,\ell\).
	If \(|K_i|=1\), then the single edge-vertex in \(K_i\) does not participate in any claws, so move on to the next index.
	If \(K_i\) does not contain a marked edge-vertex, add \(i\) to \(S\) and move on to the next index.
	Otherwise, \(\Delta(G[K_i])=4\), so apply~\cref{alg:dn4} to \(G[K_i]\).
	If Maker wins, return that Maker wins on \(\incG{G}\).
	Else, move on to the next index.
	If \(S=\emptyset\) after the loop, output ``Breaker''.
	Otherwise, return \(K_{i_1},\dots,K_{i_k}\) where \(S=\{i_1,\dots,i_k\}\).
	
	By~\cref{thm:claw_con_alg},~\cref{alg:dn4} runs in time \(\Oh{|V(G[K_i])|}\) for all \(1\leq i\leq\ell\).
	The sets \(K_1,\dots,K_\ell\) form a partition of \(E(G)\), so \(\sum_{i=1}^\ell|V(G[K_i])|\leq2|V(G)|\).
	All other steps can be executed in time \(\Oh{|V(G)|}\), so the overall running time is \(\Oh{|V(G)|}\).
	
	We now show the remaining properties of the algorithm.
	After iterating over the indices, Breaker wins on \(\incG{G[K_i]}\) for every \(i\in\{1,\dots,\ell\}\setminus S\) because the indices not in \(S\) either correspond to isolated edges or were determined to be Breaker's win by \cref{alg:dn4}.
	Thus, Maker wins on \(\incG{G}\) if and only if there exists \(1\leq j\leq k\) such that Maker wins on \(\incG{G[K_{i_j}]}\).
	In particular, Breaker wins on \(\incG{G}\) if \(S\) is empty.
	The connected components are pairwise disjoint subsets of \(E(G)\).
	Let \(i\in S\).
	As \(K_i\) is a connected component of \(\incG{G}\), \(G[K_i]\) is claw-connected.
	Additionally, \(K_i\) cannot be an isolated vertex and only contains unmarked edges, i.e.~edges incident to vertices of degree at most~3.
	Because every hyperedge in \(K_i\) corresponds to a claw, \({\Delta(G[K_i])=3}\).
\end{proof}

We obtain an algorithm for the unbiased claw game on general graphs running in time \(\Oh{|V(G)|^3}\) by combining \cref{cor:decomp_alg} with the following result:
\begin{lemmaq}[{\cite[Theorem~4.12]{DBLP:journals/corr/abs-2607-18569}}]\label{lem:cubic_alg}
	There exists an algorithm that, given a linear hypergraph \(H\) of rank at most~3, determines the winner of the Maker-Breaker game played on \(H\) in time \(\Oh{|V(H)|^3+|V(H)|^2|E(H)|}\).
\end{lemmaq}
\begin{corollary}
	There exists an algorithm that, given a graph \(G\), determines the winner of the unbiased claw game played on \(G\) in time \(\Oh{|V(G)|^3}\).
\end{corollary}
\begin{proof}
	The algorithm proceeds as follows:
	First, apply the algorithm of~\cref{cor:decomp_alg}.
	If the winner of the game on \(\incG{G}\) was already determined, return the winner.
	Otherwise, pairwise disjoint edge sets \(E_1,\dots,E_\ell\) are returned.
	For \(1\leq i\leq\ell\), use \cref{lem:cubic_alg} to determine the winner of the game on \(\incG{G[E_i]}\).
	Output ``Maker'' if Maker wins any of these games, otherwise output ``Breaker''.
	
	For \(1\leq i\leq\ell\), \(|V(\incG{G[E_i]})|=|E_i|\) and \(|E(\incG{G[E_i]})|\leq2|E_i|\), so the running time of the algorithm of~\cref{lem:cubic_alg} is \(\Oh{|E_i|^3}\).
	If \(G\) has more than \(2|V(G)|\) edges, the algorithm of~\cref{cor:decomp_alg} finds that Maker wins, so the algorithm of~\cref{lem:cubic_alg} is only used when \(|E(G)|\leq2|V(G)|\).
	As \(\sum_{i=1}^\ell|E_i|^3\) with the constraint \(\sum_{i=1}^\ell|E_i|\leq|E(G)|\) is maximized when \(E_1=E(G)\), the overall running time is \(\Oh{|V(G)|^3}\).
\end{proof}

This improves on both the algorithm of~\cite{DBLP:journals/dam/DucheneGINOPS25}, which is only an FPT-algorithm, and the algorithm of Galliot, Gravier and Sivignon~\cite{galliot2025makerbreakersolvedpolynomialtime}, which runs in time \(\Oh{|V(G)|^{16}}\).
The existence of a linear-time algorithm on general graphs is left as an open question for future work.

\section{Lower Bound on the Query Complexity}\label{sec:query_lb}
We complement our linear-time algorithm with a linear lower bound on the query complexity of deciding the unbiased claw game in the adjacency-array model, thereby showing optimality of our algorithm in this model.
The conditions in \cref{lem:two_deg_4,thm:degree_4_chara} show that the unbiased claw game is highly related to questions of connectivity.
Intuitively, connectivity cannot be decided without inspecting every vertex at least once.
We now make this intuition precise.

The adjacency-array model was used in~\cite{DBLP:journals/siamcomp/DurrHHM06} to analyze the quantum query complexity of connectivity.
In this model, a computationally unbounded agent is given the degrees \(d_1,\dots,d_n\) of the vertices and, for every vertex \(v_i\), an array \(f_i:[d_i]\to[n]\) of the neighbors of \(v_i\).
The value \(f_i(j)\) is the index of the \(j\)-th neighbor of vertex \(v_i\) according to some fixed ordering of the edges.
Only queries to the arrays are counted.
The degree-sequence is not counted as part of the queries.

\begin{figure}[t]
	\centering
	\hspace{15mm}
	$\vcenter{\hbox{\includegraphics[scale=0.8]{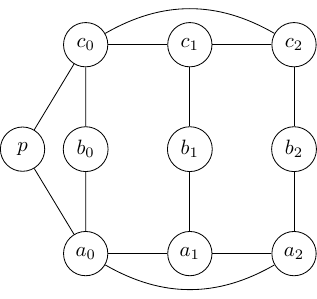}}}$
	\hspace{1cm}
	$\vcenter{\hbox{\includegraphics[scale=0.8]{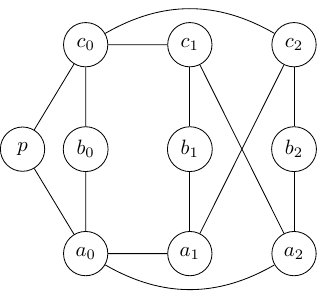}}}$
	\hspace{15mm}
	\caption{The graph \(G_2\) (left) and \(G_2^{1}\) (right) as defined in the proof of \cref{thm:lower_bound}.}
	\label{fig:query_graphs}
\end{figure}

\begin{theorem}\label{thm:lower_bound}
	In the adjacency-array model, every agent correctly deciding the unbiased claw game requires at least \(\Omega(|V(G)|)\) queries in the worst case.
\end{theorem}
\begin{proof}
	For \(2\leq k\in\mathbb{N}\), we construct a graph \(G_k\) with \(4+3k\) vertices on which any agent needs to query the presence of at least \(k+1\) edges to correctly decide the outcome of the game.
	For values of \(n\geq10\) which do not have the form \(n=4+3k\), graphs with this property are obtained by padding the graphs \(G_k\) with isolated vertices.
	Let \(2\leq k\in\mathbb{N}\) be fixed.
	The graph \(G_k\) is defined as follows:
	\begin{align*}
		V(G_k)=\,&\{p,a_0,\dots,a_k,b_0,\dots,b_k,c_0,\dots,c_k\}\\
		E(G_k)=\,&\{pa_0,pc_0,a_0a_k,c_0c_k\}\cup\{a_ia_{i+1},c_ic_{i+1}:0\leq i<k\}\cup\{a_ib_i,b_ic_i:0\leq i\leq k\}
	\end{align*}
	In the following, we view indices modulo \(k+1\), e.g.~\(a_0=a_{k+1}\).
	Let \(0\leq i\leq k\) be fixed and define the graph.
	The graph \(G_k^i\) is defined as follows:
	\begin{align*}
		V(G_k^i)= V(G_k) && E(G_k^i)=\left(E(G_k)\cup\{a_ic_{i+1},a_{i+1}c_i\}\right)\setminus\{a_ia_{i+1},c_ic_{i+1}\}
	\end{align*}
	As an example, the graphs \(G_2\) and \(G_2^{1}\) are displayed in \cref{fig:query_graphs}.
	
	Let \(A\) be an agent correctly deciding the unbiased claw game.
	By \cref{thm:degree_4_chara}, Breaker wins on \(\incG{G_k}\).
	By \cref{lem:two_deg_4}, Maker wins on \(\incG{G_k^i}\).
	Thus, \(A\) must distinguish these two graphs.
	The graphs \(G_k\) and \(G_k^i\) have the same degree sequence, so they cannot be distinguished by degrees alone.
	The only differences between the graphs are the edges in \(\{a_ia_{i+1},c_ic_{i+1},a_ic_{i+1},a_{i+1}c_i\}\), so \(A\) must query the presence of at least one of these four edges.
	As this holds for any \(0\leq i\leq k\), at least one quarter of the edges 
	\begin{equation*}
		\{a_ia_{i+1},c_ic_{i+1},a_ic_{i+1},a_{i+1}c_i:0\leq i\leq k\}
	\end{equation*}
	must be queried.
	Consequently, \(A\) queries at least \(k+1\) edges.
\end{proof}

\section{Open Problems and Future Work}
We have given a linear time algorithm for the unbiased claw game on board graphs with maximum degree different from three.
Can a linear running time also be achieved for graphs with maximum degree three?
By \cref{cor:decomp_alg}, such an algorithm would settle the complexity of the problem.

\Cref{thm:degree_4_chara} is similar to~\cite[Theorem~4.17]{DBLP:journals/corr/abs-2607-18569}.
Both use a threatening structure to improve the running time of algorithms for different Maker-Breaker games.
Is there a unified view of this technique applicable to 3-uniform hypergraphs?
Such a result could lead to an improved algorithm for the Maker-Breaker game on hypergraphs with rank at most three.

\clearpage
\bibliographystyle{elsarticle-num}
\bibliography{k13}
\end{document}